\documentclass[11pt]{article}

\usepackage[utf8]{inputenc}
\usepackage{lipsum}
\usepackage[space]{grffile}
\usepackage{graphicx}
\usepackage{url}
\usepackage{amsmath}
\usepackage{amssymb}
\usepackage{upgreek}
\usepackage{multirow}
\usepackage{verbatim}

\usepackage{algorithm}
\floatstyle{plaintop}
\restylefloat{algorithm}
\usepackage{algpseudocode}

\usepackage{mathtools}

\usepackage{tikz}
\usepackage{pgfplots}

\usepackage{tabularx}
\usepackage{esvect}
\usepackage{array, boldline, makecell, booktabs}\newcolumntype{\$}{>{\global\let\currentrowstyle\relax}}
\newcolumntype{^}{>{\currentrowstyle}}

\newcolumntype{?}{!{\vrule width 1pt}}
\usepackage{float}  
\usepackage{tikz,tikz-3dplot}
\usetikzlibrary{calc,matrix,fit}

\usetikzlibrary{positioning,arrows.meta,quotes}
\usetikzlibrary{shapes,snakes}
\usetikzlibrary{bayesnet}
\tikzset{>=latex}

\definecolor{barBlack}{HTML}{000000}
\definecolor{barGrey}{HTML}{808080}

\newcommand{\Z}{\mathbb{Z}}
\newcommand{\R}{\mathbb{R}}

\newcommand{\known}[1]{\mbox{$\sigma(#1)$}}
\newcommand{\unknown}[1]{\mbox{$\overline{\sigma}(#1)$}}

\newcommand{\coeffsubtensorScaled}{S}

\newcommand{\coeffsubtensorAdd}{T}
\newcommand{\valpha}{\vec{\alpha}}

\newcommand{\coefsubtensor}{\coeffsubtensorScaled_{k,i}}

\newcommand{\matrixsubtensor}{\coeffsubtensorScaled_k}

\newcommand{\MCA}{\mbox{MCA}}
\newcommand{\CSA}{\mbox{CSA}}

\newtheorem{theorem}{Theorem}[section]
\newtheorem{corollary}{Corollary}[theorem]

\newtheorem{proof}{Proof}[theorem]

\newtheorem{definition}{Definition}[section]

\newcommand{\SCCA}{\mbox{SCCA}}

\begin{document}

\title{\bf An Admissible Shift-Consistent Method for Recommender Systems}
\author{{\bf Tung D.\ Nguyen} and {\bf Jeffrey Uhlmann}\vspace{4pt} \\
Dept.\ of Electrical Engineering and Computer Science\\
University of Missouri - Columbia}
\date{}
\maketitle

\begin{abstract}
ABSTRACT: In this paper, we propose a new constraint, called shift-consistency, for solving matrix/tensor completion problems in the context of recommender systems. Our method provably guarantees several key mathematical properties: (1) satisfies a recently established admissibility criterion for recommender systems; (2) satisfies a definition of fairness that eliminates a specific class of potential opportunities for users to maliciously influence system recommendations; and (3) offers robustness by exploiting provable uniqueness of missing-value imputation. We provide a rigorous mathematical description of the method, including its generalization from matrix to tensor form to permit representation and exploitation of complex structural relationships among sets of user and product attributes. We argue that our analysis suggests a structured means for defining latent-space projections that can permit provable performance properties to be established for machine learning methods.
\end{abstract}

\begin{footnotesize}
\begin{quote}
\end{quote}
\end{footnotesize}

\maketitle
\section{Introduction}
Previous matrix/tensor completion methods have been defined as an optimization problem involving minimization of a particular norm of the completed matrix or tensor \cite{Recht, l2norm1, l2norm2, l2norm3, l2norm4, l2norm5, lineartime1, cv-kajo1, cv-kajo2, cv-SFO2}. One variant uses rank minimization to produce a decomposable (factorizable) result that can simplify problem complexity for subsequent operations \cite{szl, hardt, cv-background1, cv-object, cv-visual, tnnr1, tnnr2, cv-structure, low-tt-rankness}. In most cases, these generic approaches that are not defined to enforce particular application-specific properties, e.g., for recommender systems.

In our previous work \cite{acmrs}, a basic admissibility criterion was defined for recommender systems that demands that recommendation rankings be consistent with user rankings whenever there is unanimity with respect to a subset of products. For example, if all users have rated products A and B, and every user has rated/preferred A over B, then the recommender system (RS) should provably never recommend product B over product A to a new user\footnote{The actual consensus-order criterion from \cite{acmrs} is somewhat more general in that it applies also to recommendations to existing users when there are no users who have not rated both A and B, i.e., every user has rated either both or neither of the two products.}. The motivation for establishing such a simple and intuitive criterion is to provide a means for dismissing black-box systems that inherently cannot be trusted to operate fairly. Specifically, if such a system does not satisfy this basic criterion and may recommend B over A, then clearly it is operating in a way that does not reflect user sentiments. And if its recommendations do not reflect user sentiments, then what do they reflect? If, for example, it reflects a promotional payment by the producer of product B, then that fact may be of interest to users and regulators.      

The necessity for a rigorous RS admissibility criterion is to impose limits on the extent to which a system can be gamed/manipulated. The consensus-order criterion of \cite{acmrs} achieves this in the sense that such manipulations cannot lead to a given product being leapfrogged in recommendation rank over another product that is clearly preferred by the user base. We note, however, that it is also possible to game {\em the admissibility criterion itself}. Specifically, a system may include a special-case mechanism that enforces the consensus-order criterion -- {\em but only in those cases for which its conditions hold}. If this can be achieved, then the system would be able to narrowly satisfy the admissibility criterion in a purely technical sense but still retain wide latitude in all other cases to generate recommendations that would violate almost all reasonable notions of fairness. This motivates us to suggest that the criterion be augmented to require that it be satisfied generically in the sense that the same recommendation process is applied both when the admissibility conditions hold and when they do not\footnote{We also propose consideration of a more relaxed ``full-disclosure'' admissibility criterion that guarantees consensus-order except for recommendations that are given with an associated disclaimer stating the reason it violates the criterion, e.g., that it is a paid advertisement.}.  

In \cite{acmrs}, we developed a unit-consistent (UC) RS framework that provably satisfies the consensus-order admissibility criterion. It achieves this indirectly by guaranteeing that the rank ordering of recommendations is invariant to the scaling of a users set of ratings. While this framework was the first to be proven to formally satisfy the RS admissibility criterion, this does not in any way imply it is the best possible admissible approach for all RS applications, we discussed its potential sensitivity to discretization artifacts that are introduced when user ratings are constrained to integer values in a small range, e.g., 1 to 5 for the MovieLens datasets \cite{Movielens100k+1M}. In this paper, we develop a different RS framework that is provably admissible but is based on a different consistency constraint: shift consistency\footnote{Throughout this paper, we will attempt to make our shift-consistent results and analysis follow the unit-consistent results of \cite{acmrs} as closely as possible to permit direct comparison by the reader.}. Although it should be expected that all admissible systems will exhibit similar near-optimal performance simply because they each satisfy properties that are fundamental to the problem domain, we provide empirical evidence that our SC framework may offer better resilience to discretization artifacts in range-restricted RS applications.  

The structure of the paper is as follows: We begin with a brief background on recommender systems and summarize the prior-art UC framework. We show that this framework can be adapted to enforce a different consistency property, namely, shift consistency (SC). We then formally prove that the resulting SC recommender-system framework satisfies the required order-consensus admissibility criterion. Empirical results are then provided to compare performance of the UC and SC methods on the strongly range-restricted MovieLens dataset. We conclude with a discussion of how the different consistency conditions underpinning the two approaches are suggestive that similar conditions potentially may be imposed on machine learning/AI systems that otherwise would not be amenable to the proving of desired fairness and/or robustness properties.

\section{Background}

Recommender systems are essentially matrix/tensor completion algorithms that impute values for absent entries in a given table/matrix or tensor. In the RS context, filled entries are interpreted as user ratings of products, and each absent/unfilled entry represents a product that has not been rated by a particular user. For example, given a matrix $A$ with rows containing user ratings of products, then the rating of product $j$ by user $i$ can be notated as $A_{ij} = A(i,j)$ for $(i,j) \in \known{A}$. We define the list of known entries in $\known{A}$ as $A_{r}$, and the list of absent entries in $\unknown{A}$ as $A_{nr}$. The recommender system is then expected to impute values for entries in $A_{nr}$ based on the given entries in $A_r$. 

The unit-consistent RS framework of \cite{acmrs} produces a unique completion of a given table in a way that provides consistency with respect to positive scalings of the rows and/or columns. For example, assume that the value imputed for entry $(i,j)$ is $x$. If the same UC process is applied after row $i$ is scaled by positive value $s_i$, and column $j$ is scaled by positive value $s_j$, then the imputed value for entry $(i,j)$ will be $s_2s_2 x$, i.e., it is the same result but now given in the new scaled units of row $i$ and column $j$. The motivation for unit consistency is the assumption that each user implicitly applies their own personal unit(s) of measure when rating different products. Therefore, RS-predicted ratings for a particular user should be given in that user's individual units of preference. 

To get an intuitive feel for what the UC constraint enforces, consider a highly specialized example involving users Alice and Bob in which each of Alice's ratings is exactly 10\% higher than that of Bob for films they have both rated, and they have both rated exactly the same set of films. Under these conditions, the UC predicted rating (recommendation) for a particular new film to Alice would be 10\% higher than the predicted rating of that same film for Bob. In other words, the respective ratings given to Alice and Bob are consistent with the implicit units of measure each has applied in their prior ratings of films. Another way to interpret the situation is that the relative film preferences of Alice and Bob are identical in the sense that if Bob prefers film $X$ over film $Y$, then Alice has the same preference even though she consistently gives ratings that are 10\% higher than those of Bob. 

Remarkably, it has been shown that unit consistency alone is sufficient to guarantee a unique solution that can be evaluated with optimal computational efficiency \cite{acmrs}. It also permits other properties to be identified and established, one of which is a notion of {\em fairness} in the sense that a user cannot scale her set of ratings in a way that increases the influence of her relative preferences on recommendations to other users. This is because UC recommendations to other users are invariant with respect to the scaling of any user's set of ratings.

Arguably the most important property deriving from UC is consensus ordering, which has been given as an admissibility criterion for any proposed RS framework. However, unit consistency is not necessary for admissibility because other consistency conditions can also form the basis of an admissible RS. In the next section, we introduce shift consistency (SC) as such a basis.

\section{SC Completion Framework}

Ultimately, the choice to use unit consistency as a basis for RS predictions may seem to come down to a single question. In the case of the earlier example, where for any film $X$ rated by both Alice and Bob it is the case that Alice's rating is 10\% higher than Bob's rating, so if Bob rates a new film $Y$ as 70/100, should the RS be expected to predict that Alice is likely to give it a rating of 77/100, i.e., 10\% higher than Bob's rating? This is certainly not an {\em un}reasonable expectation, but it is not the only possible {\em reasonable} expectation.

Consider a variant of the running example in which the {\em difference} between Alice's rating and Bob's rating for each film is exactly 5, e.g., if Alice rates a particular film as 85/100 then Bob's rating is 80/100. Thus, instead of a multiplicative/scale relationship between their respective sets of ratings, there is now a fixed additive relationship in the sense that each of Bob's ratings is a shift/translation of Alice's corresponding rating. Under these conditions, if Bob rates a new film $Y$ as 70/100, should the RS be expected to predict that Alice is likely to give $Y$ a rating of 75/100? This expectation is no less reasonable than what was expected in the UC case. 

Our motivation for examining shift consistency is not based on whether or not it is more reasonable in some sense than unit consistency. That is a question that demands an application-specific answer. Our motivation is based on an observation in \cite{acmrs} that the multiplicative/scale properties of UC may exhibit undesired sensitivity to discretization when ratings are confined to a small set of integer values. This can be understood by considering the extent to which percentage relationships among ratings can be accurately discerned when ratings are restricted to, e.g., values between 1 and 10, as opposed to consistent additive differences. If it is assumed that all admissible RS systems are likely to exhibit similar performance due to the fact that they all enforce reasonable RS-related properties, then considerations such as sensitivity to discretization may represent a deciding factor when choosing among otherwise equally-reasonable candidates. 

\subsection{Definitions}
\label{CSA_section}

 Our Shift-Consistent Completion Algorithm (SCCA) is based on a canonical shifting algorithm (CSA), and its derivation closely follows that for UC in \cite{acmrs}. Given a $d$-dimensional tensor $A$, where  $d=2$ corresponds to a matrix (i.e., the conventional RS tabular formulation), we begin with the following definitions, which also follow those in \cite{acmrs}:

\begin{enumerate}

\item $A$ is a positive $d$-dimensional tensor with fixed dimensional extents $n_1, ..., n_d$. (The restriction to strictly positive, as opposed to nonnegative, entries is unnecessary for canonical shifting but will prove convenient for maintaining consistency with the common RS convention of reserving the value of $0$ to represent absent entries, and also for simplifying notation by avoiding the need to define a special symbol to designate absent entries.) 

\item $\Vec{\alpha} = \{\alpha_1, ...\,, \alpha_d\} \in \Z^d$ is a $d$-dimensional vector that specifies an entry of $A$ as $A(\Vec{\alpha})$.

\item $A_i$ is the $i^{th}$ element of the ordered set of all $k$-dimensional subtensors of $A$, where each subtensor is equal to $A$ but with a distinct subset of the $d-k$ extents restricted to $1$. If vector $\vec{\alpha}$ satisfies $A(\vec{\alpha}) \in A_i$, we write $\vec{\alpha} \in A_i$.

\item $\coeffsubtensorScaled_k$ is a vector of length equal to the number of $k$-dimensional subtensors of $A$, with $S_{k, i}$ denoting the $i^{th}$ element.

\item $\known{A}$ is the set of known entries of tensor $A$, and $\unknown{A}$ is its complement, the set of unknown entries.

\item Integer $k<d$ denotes the dimensionality of a given subtensor.

\item $[m] = \{1, 2, ...\,, m\}$ is an index set of the first $m$ natural numbers.

\end{enumerate}

\noindent The following is the first of two key definitions.

\begin{definition} - $\textup{\bf Sets of subtensors:}$ Let $V_i = \mathbb{R}^{n_i}$ and assume a $d$-dimensional tensor $A \in V_1 \otimes V_2 \otimes \cdots \otimes V_d$ and a positive integer $1 \leq k < d$. For a permutation vector $\pi = \{\pi_1,\cdots,\pi_k\} \subset [d]$ of cardinality $k$, a flattening $(\pi, \pi^c)$ of $A$ is the matrix $A_{(\pi^c, \pi)}$ defined as:
\begin{equation}
    A_{(\pi^c, \pi)} \in \left(V_{\pi_1^c} \otimes \cdots \otimes V_{\pi^c_{d-k}}\right) \otimes \left(V_{\pi_1} \otimes \cdots \otimes V_{\pi_k} \right)^{*} 
    \label{flat}
\end{equation}
where $V_{\pi_i}^{*}$ denotes the dual space of $V_{\pi_i}$. Observe that the rows of $A_{(\pi^c, \pi)}$ matrix contain all the k-dimensional subtensors corresponding to the permutation $\pi$. The set of all subtensors is denoted as $\mathcal{A}$ (with finite cardinality $|\mathcal{A}|$) is:
\begin{equation}
    \mathcal{A} = \bigcup_{\pi \subset [d]} \{A_{(\pi^c, \pi)}[r, :], 1 \leq r \leq size( A_{(\pi^c, \pi)}, 1 ) \}
\end{equation}
where $size(T, n)$ is the dimensional length of a matrix $T$ at dimension $n$. For notation convenience, a subtensor in $\mathcal{A}$ is denoted as $A_i$ for $1\leq i \leq |\mathcal{A}|$.
\label{subtensors_def}
\end{definition}
The flattening space in equation (\ref{flat}) puts the indices corresponding to the dimensions specified by $\pi$ first, and the indices corresponding to the dimensions specified by $\pi^c$ last. The rows of $A_{(\pi^c, \pi)}$ correspond to the k-dimensional subtensors of $A$, where the first $k$ indices match the dimensions specified by $\pi$, and the last ($d$-$k$) indices match the dimensions specified by $\pi^c$. The following is our second key definition.
\begin{definition} \label{scaleDef} - $\textup{\bf Shifting $k$-dimensional subtensors of a tensor:}$ 
For $d$-dimensional tensor $A$, let $S_{\pi}$ be the diagonal matrix that scales all $k$-dimensional subtensors by permutation vector $\pi\subset [d]$ with cardinality $k$ and let $S_k \coloneqq (S_{\pi})_{\pi \subset [d]}$. We define the $\pi$-shift as
\begin{equation}
   S_{\pi}\ominus_{\pi}A = A_{(\pi^c, \pi)} - S_{\pi} \otimes \mathbf{1}_{size(A_{(\pi^c, \pi)}, 1)} \quad \forall \pi \subset [d] \,.
\label{scale-by-all}
\end{equation}
where $\mathbf{1}_{n}$ is a column vector of ones of size $n$. Then the shifting operation $A'=\coeffsubtensorScaled_k \ominus_k A$ is defined as
\begin{equation}
    (S_{\pi_1}, \cdots, S_{\pi_k})\ominus_k A \coloneqq S_{\pi_1} \ominus_{\pi_1} (\cdots (S_{\pi_{k-1}} \ominus_{\pi_{k-1}} (S_{\pi_l} \ominus_{\pi_l} A))\cdots)
\end{equation}
Equivalently, for each $\vec{\alpha} \in \upsigma(A)$ 
\begin{equation}
    A'(\Vec{\alpha}) ~~ \equiv   ~~
A(\Vec{\alpha}) ~-\sum\limits_{i \text{:} \Vec{\alpha} \in A_i}\hspace{-2pt} \coefsubtensor ~.
\label{simple-alpha-shifting}
\end{equation}

\label{shiftingbyHadamard}
\end{definition}
Formula (\ref{scale-by-all}) shifts each $A_i \in \mathcal{A}$ by a coefficient from matrix $S_{\pi}$, with each coefficient $S_{k,i}$ shifting the corresponding subtensor $A_i$. This yields tensor $A'$, with a simpler equivalent formula (\ref{simple-alpha-shifting}) used for subsequent analyses. The structured shifting of Definition \ref{scaleDef} is a fundamental component of the solution for the {\em canonical shifting problem}, which transforms a nonnegative tensor to a unique shift-invariant canonical form.

\begin{algorithm}
\caption{Shift-Consistent Completion Algorithm (SCCA)}\label{alg:scca}
\begin{algorithmic}[1]
\State Input: $A$, $k$
\State Output: $A'$

\Function{SCCA}{$A, k$}
\State $S_k \gets$ \Call{CSA}{$A, k$} \Comment{Step 1: CSA process}
\State $A' \gets A$ \Comment{Step 2: Completion process}
\For{$\Vec{\alpha} \in \unknown{A}$}
\State $A'(\Vec{\alpha}) \gets \sum_{i: \vec{\alpha} \in A_i} S_{k,i}$
\EndFor
\EndFunction
\end{algorithmic}
\end{algorithm}

\begin{algorithm}[H]
\caption{Canonical Shifting Algorithm (CSA)}\label{alg:csa}
\begin{algorithmic}[1]
\State Input: $A$, $k$
\State Output: $A'$, $S_k$
\Function{CSA}{$A, k$}
\State Initialize threshold $\epsilon \in (0,1)$, $A'$ = $A$, variance $v\leftarrow 0$, and $p$ and $S_k$ as a zero vector of conformant length.
\While{$v < \epsilon$}
\For{each subtensor $A_i$ with index $i$}
\State $\rho_i = -\left[ |\upsigma(A_i)|\right]^{-1}
\sum\limits_{\Vec{\alpha} \in \upsigma(A_i)}A'(\Vec{\alpha})$

\State $A'(\Vec{\alpha}) \leftarrow A'(\Vec{\alpha}) + \rho_{i}$, \quad $v\leftarrow v+\rho_i^2$ \quad for $\Vec{\alpha} \in \upsigma(A_i)$
\State $\coefsubtensor \leftarrow \coefsubtensor + \rho_{i}$
\EndFor
\EndWhile
\EndFunction
\end{algorithmic}
\end{algorithm}

The time complexity of this algorithm is $O\left(|\upsigma(A)| \right)$ per iteration, assuming that $d=O(1)$, which is all but necessary in any realistic practical application \cite{CSA}. Using definition \ref{shiftingbyHadamard}, the CSA algorithm is formulated by the CSP formulation for arbitrary shifting of $k$-dimensional subtensors of a given nonnegative tensor $A$.
\\\\
\textbf{Shift-Consistent Problem (CSP):} Find the $d$-dimensional tensor $A'$ and a vector $\matrixsubtensor$ such that $
A'=\coeffsubtensorScaled_k \ominus_k A$,
and each $i^{th}$-index $k$-dimensional subtensor $A'_i$ with the sum of its known entries equal to zero.
\begin{equation}
\sum_{\Vec{\alpha} \in \upsigma(A_i)} A'(\Vec{\alpha}) = 0 \quad \forall A_i \in \mathcal{A}
\end{equation}

Here we state the uniqueness solution.

\begin{theorem}{(Uniqueness of A')} There exists at most one tensor $A'$ for which there exists a vector $\matrixsubtensor$ such that the solution $\left(A',\matrixsubtensor \right)$ satisfies CSP. Furthermore, if a vector $\coeffsubtensorAdd_k$ satisfies
\begin{equation}
    \sum\limits_{i: \Vec{\alpha} \in A_i}\hspace{-4pt}  \coeffsubtensorAdd_{k, i}  ~= ~0 \qquad \forall \Vec{\alpha}  \in \upsigma(A),
\end{equation} then $\left(A', \coeffsubtensorScaled_k + \coeffsubtensorAdd_k\right)$ is also a solution.
\label{uniquess_A'}
\end{theorem}{}

The detailed formulation and relation to prove its uniqueness and existence of solution for the \textbf{CSP} problem can be derived using \cite{acmrs}.

\subsection{Uniqueness and Full Support}
Although the canonical shifted tensor is unique, the shifted vectors, or the resulting $A'$ from the $\SCCA$ process, may not be unless there are sufficient known entries to provide {\em full support}, which is now defined. 

\begin{definition} Given A, we say that tensor is {\em fully supported} if for every entry $\Vec{\alpha} \in \unknown{A}$, {\em there exists} a nonzero vector $\vec{s}$ such that $\left(\alpha_1 + \delta_1 s_1, \cdots, \alpha_d + \delta_d s_d \right) \in \known{A}$ for all $2^{d-1}$ choices of $\delta_i \in \{0,1\}$ such that $\delta_1 + \cdots \delta_d > 0$. We denote $\vec{\alpha}' = \left(\alpha_1 + \delta_1 s_1, \cdots, \alpha_d + \delta_d s_d \right) = \vec{\alpha} + \vec{\delta}\cdot \vec{s}$ and the set of such vectors $\Vec{\alpha}'$ as $H(\Vec{\alpha}, \Vec{s})$.
\label{fully-supported-tensor}
\end{definition}
Definition \ref{fully-supported-tensor} essentially means that each unknown entry forms a vertex of a hypercube with known entries at the other vertices. For example, in the matrix case, an unknown entry $(i,j)$ must have known entries at ($i$+$p$,\,$j$), ($i$,\,$j$+$q$), and ($i$+$p$,\,$j$+$q$) for some $p\neq 0$ and $q\neq 0$. The following theorem establishes the uniqueness of the recommendation or entry-completion result assuming full support.
\begin{theorem} \textbf{(Uniqueness)}
The result from $\SCCA (A, k)$ is unique, even if there are different sets of shifting vectors ${\matrixsubtensor} \leftarrow \CSA(A, k)$.
\label{uniqueness_proof}
\end{theorem}
The detailed proof is in the Appendix \ref{appendix_uniqueness}.

\subsection{Shift Consistency}
$A'$=$\CSA(A,k)$ guarantees shift-invariance with respect to every $k$-dimensional subtensor of $A'$. It then remains to show that the completion result from $\SCCA(A,k)$ is shift-consistent.

\begin{theorem} (Shift-consistency)
Given a tensor $A$ and an arbitrary conformant shifting vector $T_{k}$,  $T_{k} \ominus_k \SCCA(A, k)\,=\,\SCCA(T_{k} \ominus_k A, k)$,
where $T_k$ shifts all $k$-dimensional subtensors of $A$ (with operator $\ominus_k$ as defined in definition \ref{scaleDef}).
\label{shift-consistecy}
\end{theorem}

\begin{proof}
Let $\coeffsubtensorScaled_k \leftarrow \CSA(A, k)$, then $A' = \CSA(T_k \ominus_k A, k) = \CSA(A, k)$ for all $A$. For simplicity, we assume all unknown entries of $A'$ are assigned the value of 0, i.e., $A'(\vec{\alpha}) = 0$ for $\valpha \in \unknown{A}$. The SCCA output can then be defined as $\SCCA(A, k) = \coeffsubtensorScaled_{k}\ominus_k A'$. Now, using the uniqueness theorem \ref{uniqueness_proof}, we can subsume the shifting vector $T_k$ into $S_k$ and derive that
\begin{equation}
    T_{k} \ominus_k \SCCA(A, k) =  (\coeffsubtensorScaled_{k} + T_{k}) \ominus_k A' 
    = \SCCA(T_{k} \ominus_k A, k) \,.
\end{equation}
\end{proof} 

\subsection{Consensus Ordering Property}

We now proceed to prove that our methods provide solutions that satisfy the consensus-ordering admissibility property of \cite{acmrs}. For any recommender system, we show that the following method satisfies consensus ordering with respect to an ordering relationship of the $(d$$-$$1)$-dimensional subtensors of a given $d$-dimensional tensor $A$. 

\begin{definition} Consensus Ordering Set: Given a tensor $A$ and vector $\Vec{\alpha} \in \mathbb{Z}^{d-1}_{> 0}$, we define a order-($d-1$) tensor $A^{(n)} \in V_1 \otimes V_1 \otimes \cdots \otimes V_{d-1}$ as
\begin{equation}
    A^{(n)}(\Vec{\alpha}) \coloneqq A(\alpha_1, \cdots, \alpha_{d-1}, n), \qquad 1 \leq n \leq n_d.
\end{equation}
For a permutation vector $\gamma \subset [n_d]$ with size $D$ and tensor $A$, we define a set of known vectors $\known{\gamma}$ that preserves/follows ordering $\gamma$ in tensor $A$ {\em iff}
\begin{equation}
\begin{aligned}
\known{A^{(\gamma_i)}} &= \known{\gamma}, \quad \text{for all} \ 1 \leq i \leq D. \\
   A^{(\gamma_a)}(\vec{\alpha}) &< A^{(\gamma_b)}(\vec{\alpha}) \qquad \forall  \vec{\alpha} \in \known{ \gamma} \text{ and } 1 \leq a < b \leq D.
\end{aligned}
\end{equation}
Then the set of absent/missing vectors $\unknown{\gamma}$ satisfies $\unknown{\gamma} = \unknown{A^{(\gamma_i)}}$ for all $1 \leq i \leq D $.
\label{unanimous}
\end{definition}

Now we can formally state the following theorem.

\begin{theorem} Consensus Ordering: 
     Given a tensor $A$ and obtained result $A' = \SCCA(A, d-1)$, and permutation vector $\gamma$, $\known{\gamma} \neq \emptyset$, then any completion vector $\vec{\alpha} \in \unknown{\gamma}$ must satisfy $ A'^{(\gamma_a)}(\vec{\alpha}) < A'^{(\gamma_b)}(\vec{\alpha}) \text{ when } a < b.$
     \label{unanimous-ordering}
\end{theorem}
The proof can be found in the Appendix \ref{consensus-proof}.

\section{SC Application for Recommender System}

In this section, we examine the application of SC completion to recommender systems. Because of its particular relevance to matrix-formulated recommender-system problems, we briefly 
discuss the special case of $d=2$, $k=1$, for a given $m\times n$ matrix\footnote{Because all results in this paper are transposition consistent, we implicitly assume without loss of generality that $n\geq m$ purely to be consistent with our general use of $n$ as the variable that functionally determines the time and space complexity of our algorithms.} $A\in \R_{> 0}^{m\times n}$ with full support. For notational convenience, we define the matrix completion function $\MCA(A)$ as a special case of SCCA:
\begin{equation}
       \MCA(A) ~\equiv ~ SCCA(A,1) ~~ \mbox{for $d=2$}\, .
\end{equation}
In this case, the shift-consistency property can be expressed as $R \otimes \mathbf{1}_n + \MCA(A) +  \mathbf{1}_m \otimes C  = \MCA(R\otimes \mathbf{1}_n + A + \mathbf{1}_m \otimes C)$ for vectors $R$ and $C$. The time and space
complexity for $\MCA(A)$ is $O(|\upsigma (A)|)$.

Theorem \ref{unanimous-ordering} and Definition \ref{unanimous} provide a means to specify the recommendation method in the following definition.
\begin{definition}
Denote $RS(\Vec{\alpha}) = A'(\Vec{\alpha})$ as the recommendation result for vector position $\Vec{\alpha} \in \mathbb{Z}^{d}$ from tensor $A$ with $A' = \SCCA(A, d-1)$. 
\end{definition}

\subsection{Consensus Ordering}
Following Definition \ref{unanimous} and Theorem \ref{unanimous-ordering}, we formally state the consensus ordering property in the context of a matrix and 3-dimensional tensor. 

\begin{corollary} 
Given a matrix $A \in \mathbb{R}^{m \times n}$, $\MCA(A)$, set of products $P \subseteq [n]$, and set of users $U \subseteq [m]$. We have the following statements:
\begin{enumerate}
    \item Given a permutation by products $\gamma_P$ and nonempty set $\known{\gamma_P}$. Then the recommendation result $RS(u, p)$ for any user $u \in \unknown{\gamma_P}$ and product $p\in \gamma_P$ follows consensus ordering $\gamma_P$. 
    \item Given a permutation by users $\gamma_U$ and nonempty set $\known{\gamma_U}$. Then the recommendation result $RS(u, p)$ for any product $p \in \unknown{\gamma_U}$ and user $u \in \gamma_U$ follows consensus ordering $\gamma_U$. 
\end{enumerate}
\end{corollary}

For a 3D tensor, we have the following interpretation.
\begin{corollary} Given a $3$-dimensional tensor $A$. For the recommendation result $\SCCA(A, 2)$, we have the following statements:

\begin{enumerate}
    \item  Given a permutation by users $\gamma_U$ and nonempty set $\known{\gamma_U}$. Then the recommendation result $RS(u, \alpha, p)$ for attribute-product vector $(\alpha, p)\in \unknown{\gamma_U}$ and each user $u \in \gamma_U$ follows consensus ordering $\gamma_U$.

    \item  Given a permutation by products $\gamma_P$ and nonempty set $\known{\gamma_P}$. Then the recommendation result $RS(u, \alpha, p)$ for user-attribute vector $(u, \alpha)\in \unknown{\gamma_P}$ and each product $p\in \gamma_P$ follows consensus ordering $\gamma_P$.
    
    \item  Given a permutation by attributes $\gamma_{\Gamma}$ and nonempty set $\known{\gamma_{\Gamma}}$. Then the recommendation result $RS(u, \alpha, p)$ for user-product vector $(u, p)\in \unknown{\gamma_{\Gamma}}$ and each attribute $\alpha \in \gamma_{\Gamma}$ follows consensus ordering $\gamma_{\Gamma}$.
\end{enumerate}
\label{3d}
\end{corollary}
This corollary extends the 
symmetry property of theorem \ref{unanimous-ordering} with respect to any dimensions of preference in the tensor: users, products, and attributes, i.e., whether attributes are defined with respect to users or products. The generalization from a 2-dimensional matrix to a 3-dimensional tensor even further exploits this symmetry of ordering with respect to the choice of spatial label on different coordinates. 

This corollary establishes how the consensus ordering property with respect to the same index label in the original tensor is preserved in the output tensor. Specifically, the ordering follows directly from theorem \ref{unanimous-ordering} whenever there exists unanimity among attributes on the first $d$-$1$ dimensions with respect to the last coordinate.

In the following section, we conclude our analysis of the framework with consideration of how shift consistency provides a precise understanding of fairness in terms of the impact of each user's ratings on the recommendations generated by a SC recommender system.

\subsection{Fairness and Robustness}
\label{fairness}

We note that the SC criterion leads to recommendations that are invariant with respect to a global shifting of any row or column of the RS matrix, so an SC-based recommender system is less susceptible to biases in the sense that a user cannot add a constant to each rating to influence the ``weight'' given to it by, e.g., a squared-error minimizing RS that gives more weight to higher (larger magnitude) ratings than to lower ones. This invariance of SC to such shifts can be formalized as follows:

\begin{theorem}
Let $A$ be a tensor, $RS$ be the recommender system obtained from $SCCA(A, d-1)$, and $RS(u, p)$ denote the recommendation result for user $u$ and product $p$ in $RS$. For a given user $u$ and a subtensor $A_i$ representing their subset of ratings for different products, let $T_k$ be a shifting vector applied by user $u$ to the $RS$ such that $T_{k, i} \neq 0$ and $T_{k, i'} = 0$ for $i' \neq i$, resulting in a new recommender system $RS_{new}$. Then for any user $u' \neq u$ and product $p$, we have $RS_{new}(u', p) = RS(u', p)$.
\end{theorem}

More generally, our belief is that increasing the number of {\em reasonable} constraints tends to decrease a system's latitude to exhibit {\em un}reasonable behaviors. This is not only desirable for promoting robustness, it is also relevant to fairness, which in part is defined by user perception. For users from marginalized minority groups, there will always be concerns about exactly {\em what} a learning-based system has learned from information that by definition over-represents the majority. The transparency of SC (or UC) constraints therefore may tend to promote a greater sense of inclusivity for such users. In fact, the order-consensus admissibility criterion is explicitly defined in terms not of a majority agreement but on condition of unanimous agreement.   

\section{Empirical Results}

In this section we discuss various issues relating to the evaluation of RS performance. We then provide empirical corroboration for theoretically-proven SC properties relating to fairness and robustness. Our tests are performed using the following MovieLens datasets in Table \ref{table:movielens}:


\begin{table}[h]
\centering
\caption{Statistics of 3 MovieLens datasets}
\label{table:movielens}
\begin{tabular}{ccccc}
\hline
\textbf{Dataset} & \textbf{\# Users} & \textbf{\# Products} & \textbf{\# Ratings} & \textbf{Sparsity level} \\ \hline
MovieLens 100k   & 943  & 1,682  & 100,000  & 93.7\%   \\ 
MovieLens 1M     & 6,040 & 3,900 & 1,000,209 & 95.5\%  \\ 
MovieLens 10M & 69,878 & 10,677 & 10,000,054 & 98.4\% \\ \hline
\end{tabular}
\end{table}

\subsection{RMSE and MAE Tests}

We begin with tests comparing the new SC method to the prior UC method according to simple root-mean-squared error (RMSE) and maximum absolute error (MAE). Our goal is to assess the hypothesis that SC will be less sensitive to the heavy discretization imposed by the limiting of MovieLens ratings to the integers 1 to 5. As can be seen in Figures \ref{100k-rmse}-\ref{10m-mae},
SC exhibits slightly better performance than UC on each of the datasets according to both measures, which provides limited support for the discretization hypothesis as an explanation for that small difference. 

\begin{figure}
\includegraphics[width=4.5in, height=2.5in]{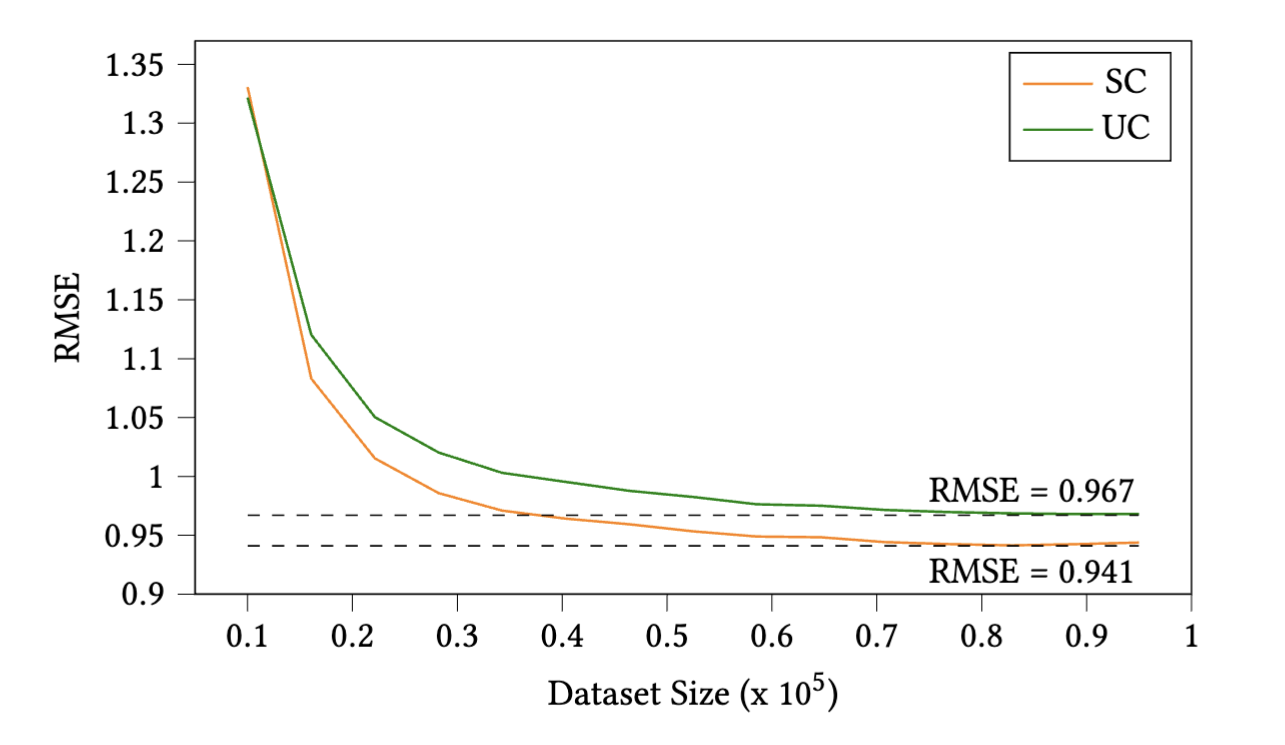}
\caption{This figure shows how the RMSE decreases with dataset size, where the horizontal label represents increasing fractions of the MovieLens 100k dataset. Here the RMSE of SC is $0.02$ smaller than UC.}
\label{100k-rmse}
\end{figure}

\begin{figure}
\includegraphics[width=4.5in, height=2.5in]{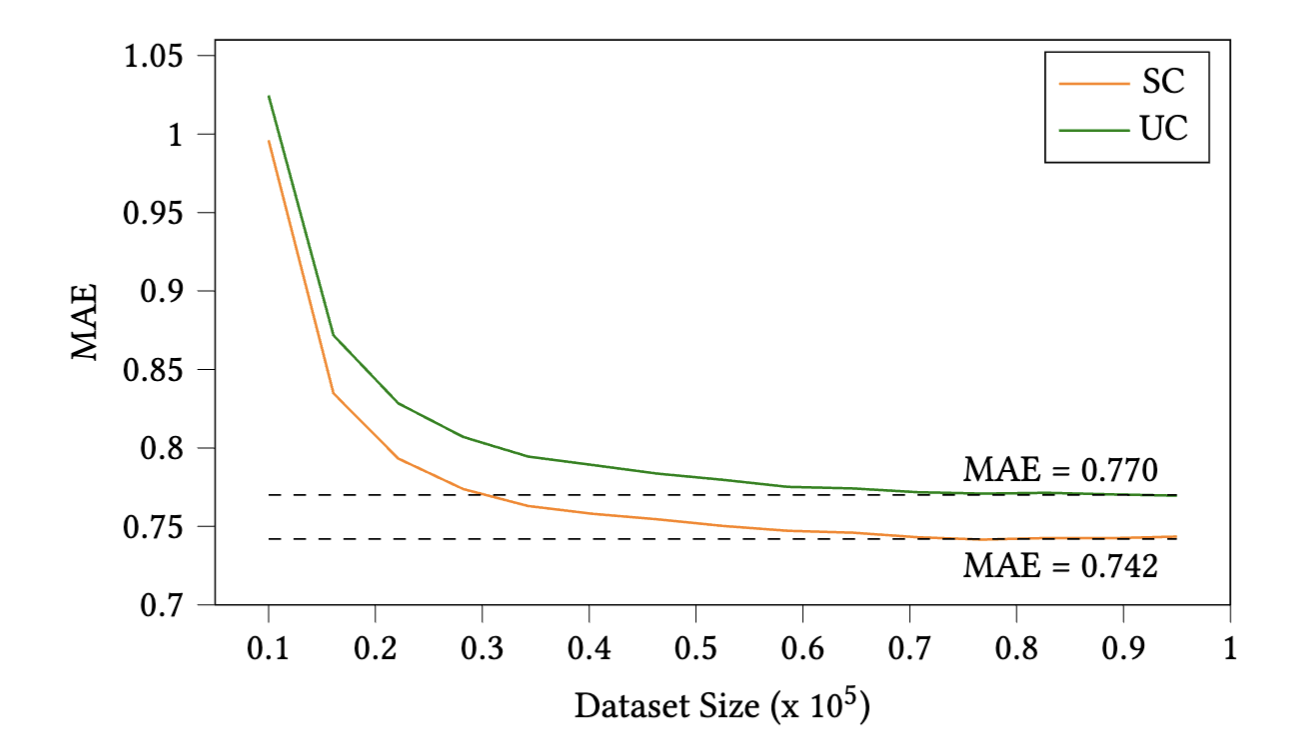}
\caption{This figure shows how the MAE decreases with dataset size, where the horizontal label represents increasing fractions of the MovieLens 100k dataset. Here the RMSE of SC is $0.03$ smaller than UC.}
\label{100k-mae}
\end{figure}

\begin{figure}
\includegraphics[width=4in, height=2.5in]{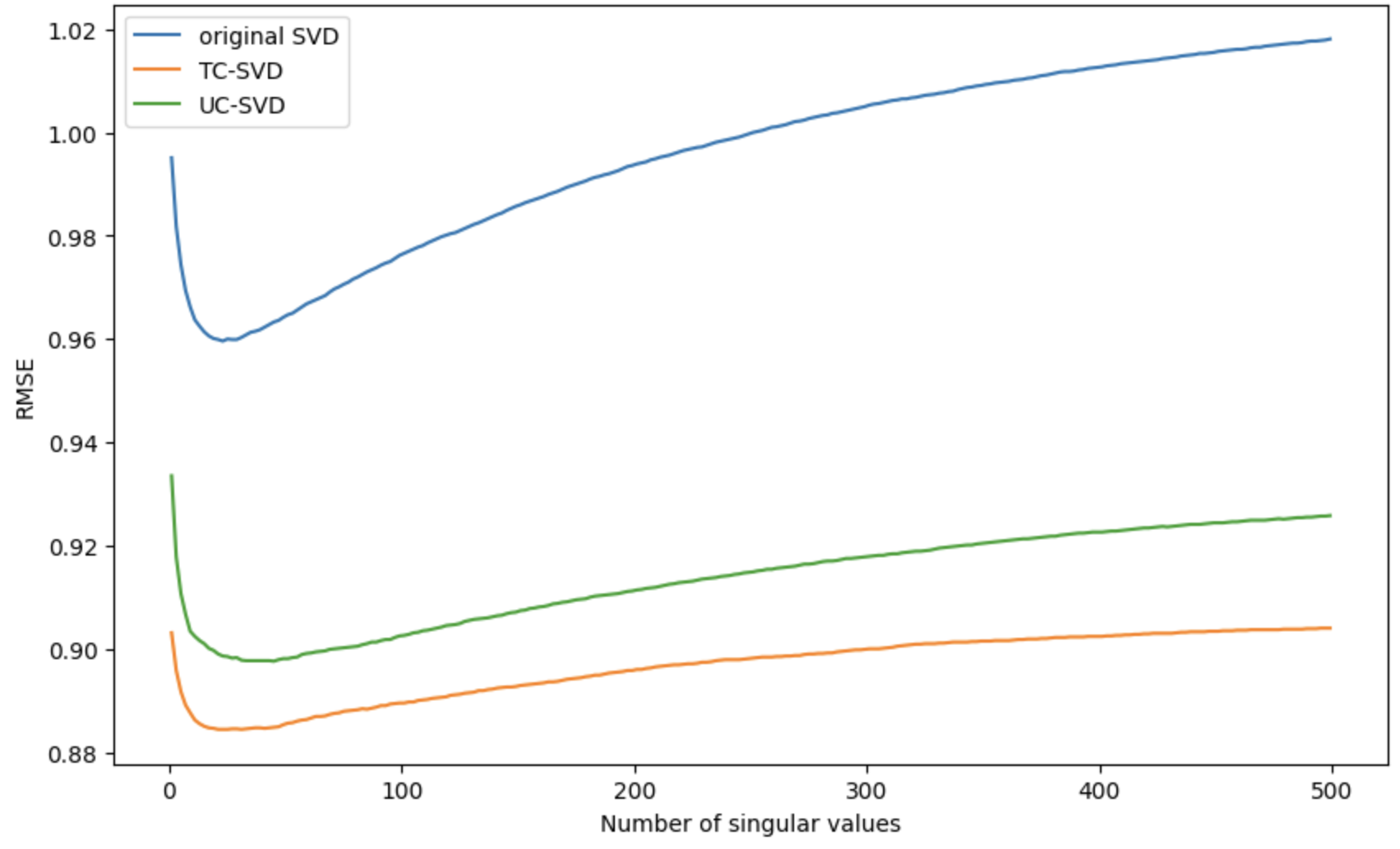}
\caption{This figure shows how the RMSE decreases with dataset size, where the horizontal label represents increasing fractions of the MovieLens 1M dataset. Here the RMSE of SC is $0.02$ smaller than UC.}
\label{1m-rmse}
\end{figure}

\begin{figure}
\includegraphics[width=4in, height=2.5in]{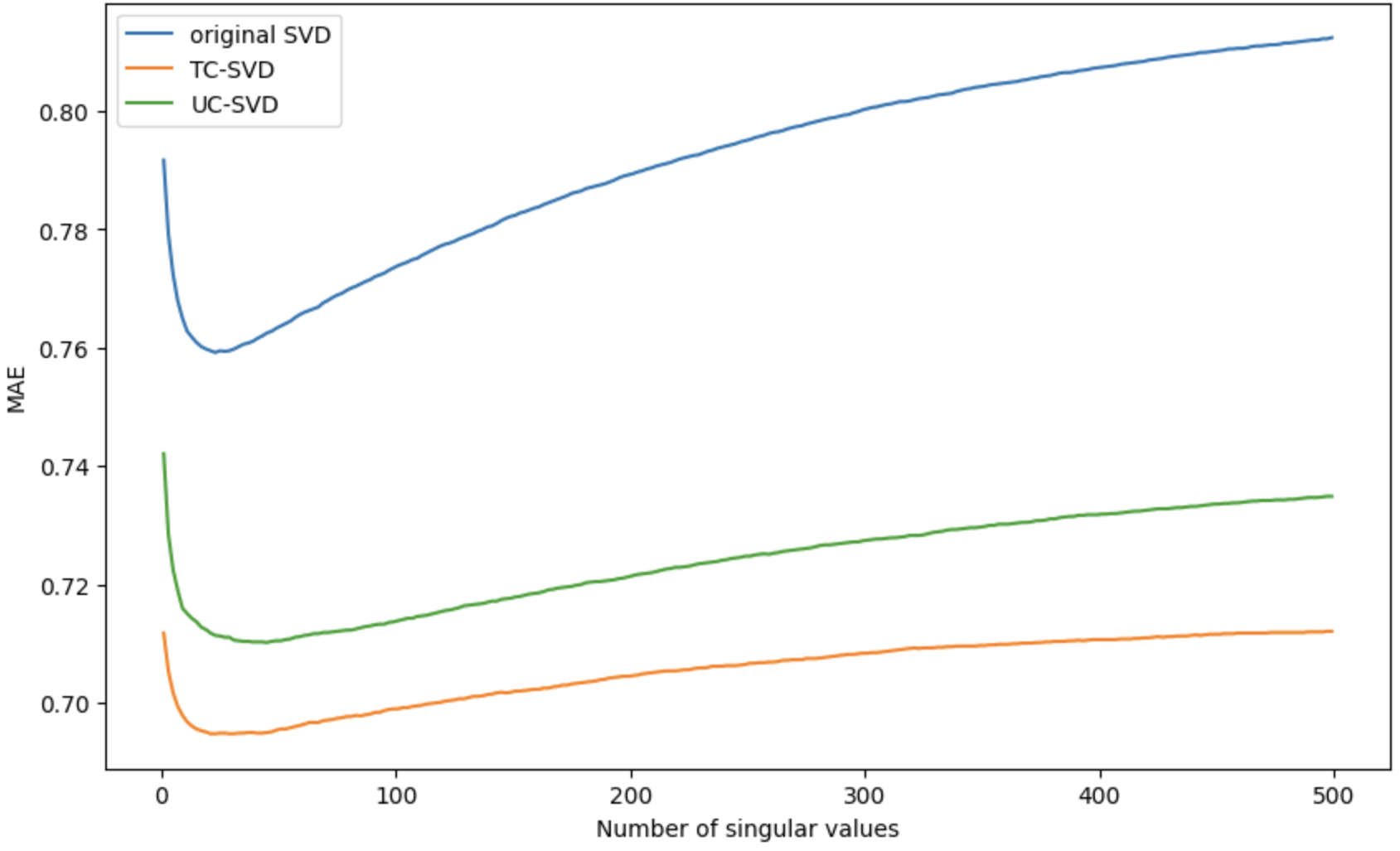}
\caption{This figure shows how the MAE decreases with dataset size, where the horizontal label represents increasing fractions of the MovieLens1M. Here the RMSE of SC is $0.03$ smaller than UC.}
\label{1m-mae}
\end{figure}

\begin{figure}
\includegraphics[width=4.5in, height=2.5in]{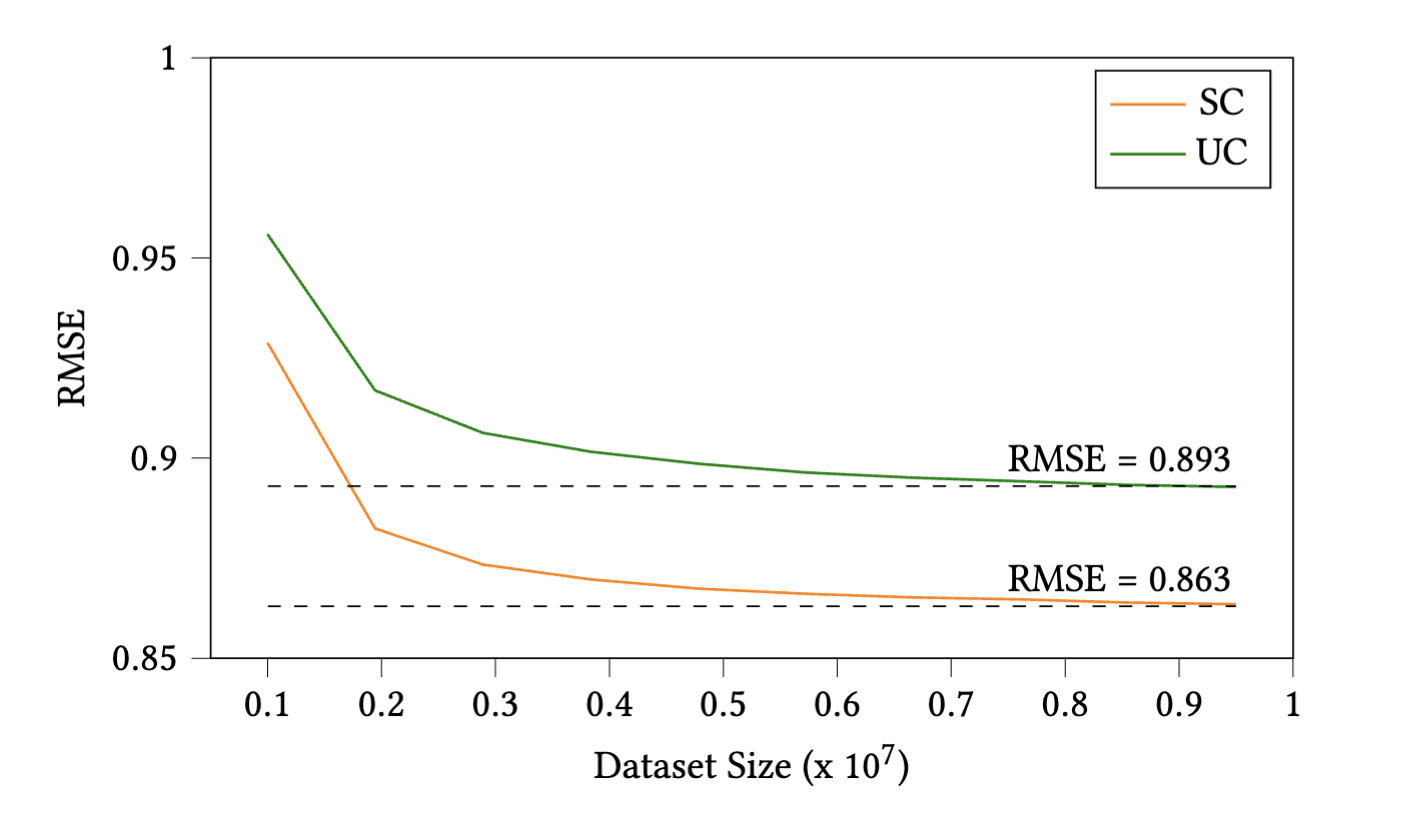}
\caption{This figure shows how the RMSE decreases with dataset size, where the horizontal label represents increasing fractions of the MovieLens 10M dataset. Here the RMSE of SC is $0.03$ smaller than UC.}
\label{10m-rmse}
\end{figure}

\begin{figure}
\includegraphics[width=4in, height=2.5in]{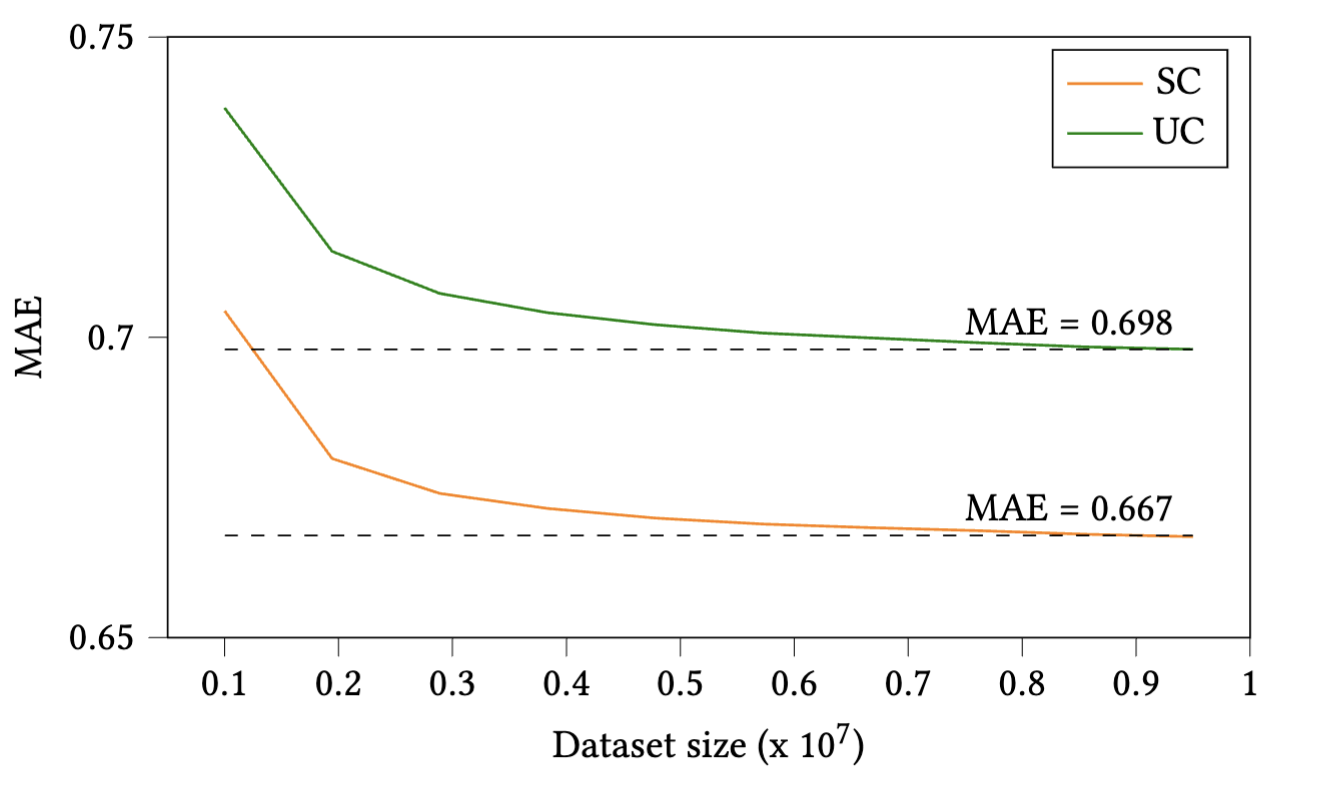}
\caption{This figure shows how the MAE decreases with dataset size, where the horizontal label represents increasing fractions of the MovieLens 10M dataset. Here the MAE of SC is $0.03$ smaller than UC.}
\label{10m-mae}
\end{figure}

\begin{figure}[t] \centering
\includegraphics[width=4in, height=2.5in]{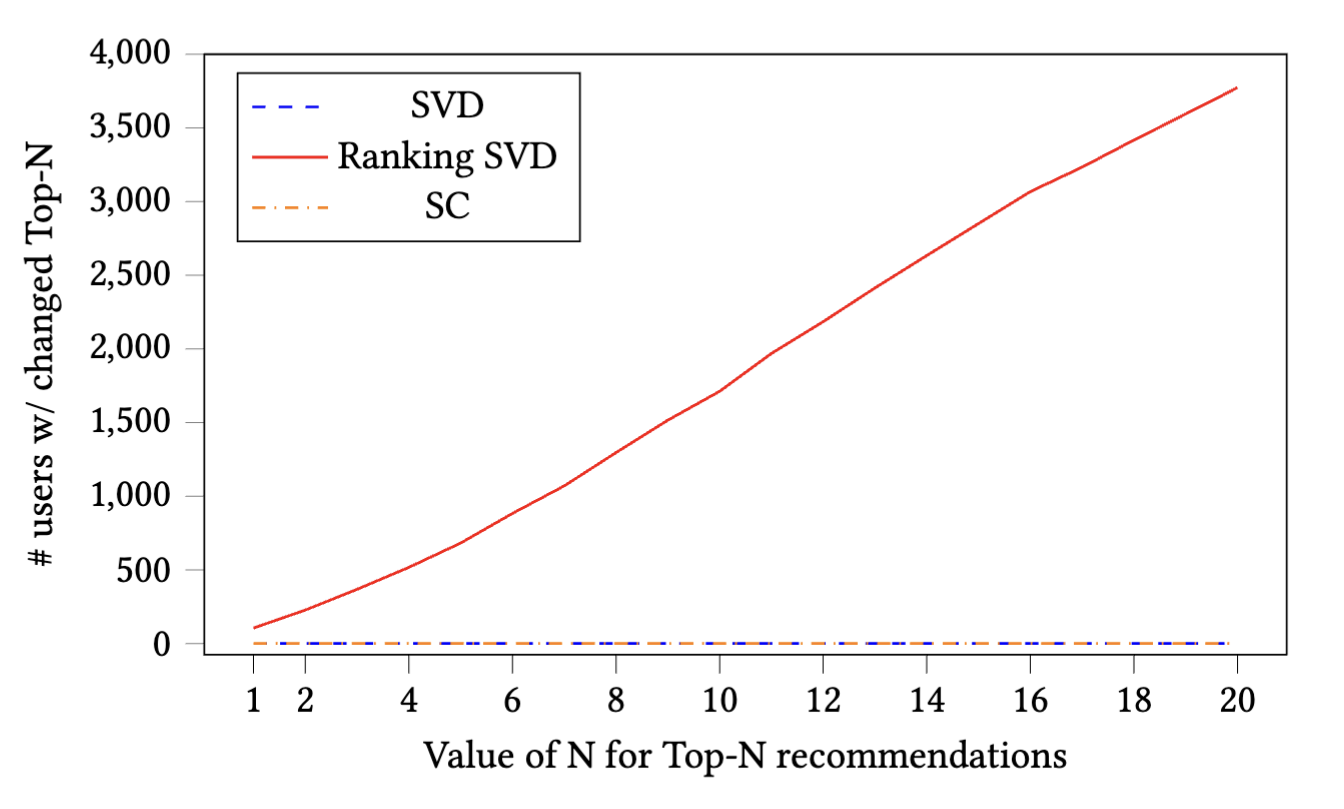}
    \caption{This figure shows the effect of adding 1 to all of the ratings of a single user in the MovieLens1M dataset so that their highest rating of 4 becomes the maximum rating value of 5. More specifically, we consider the extent to which such a shifting affects the rank ordering of recommendations made to other users. The horizontal axis represents the number of changes in the top-N, and the vertical axis represents the number of users whose top-N ranked recommendations from the system changes as a consequent of the single user shifting their ratings. For example, between 500 and 1000 users would see one or more changes to their top-10 recommendations from the system.}
\label{fair}
\end{figure}

Our final test is intended to provide empirical corroboration that the SC rank ordering of recommendations is invariant to a shift/translation of +1 applied to all of a particular user's ratings. As can be observed in Figure \ref{fair}, shifting of the user's ratings has no effect on the rank order of SC recommendations. By contrast, the rank ordering of recommendations from alternative SVD-based methods (``SVD'' \cite{firstSVD}  and ``Ranking SVD'' \cite{ranking-SVD}) exhibit dramatic changes, including even to their respective sets of Top N recommendations.  

Future work will examine and provide test results on the enforcement of UC and SC constraints on blackbox systems \cite{short}.

\section{Summary}

In this paper we have presented a shift-consistent alternative to unit consistency for constructing a provably admissible RS framework. We have argued that any formally admissible RS framework should be expected to produce {\em reasonable} results in the sense of recommendations that reflect rank-order preference information from user ratings. However, we have noted that admissible candidate RS frameworks exhibit different sensitivities to application-specific artifacts. In particular, empirical evidence was presented that shift consistency may provide greater robustness to extreme rating discretization compared to UC. 

A more significant contribution of this paper is a concrete example of a new provably admissible RS framework that complements the original UC framework. This demonstrates that multiple consistency criteria can serve as the basis for an admissible RS. This is important because the imposition of such consistency constraints on machine learning / AI systems may be necessary to render their performance properties amenable to rigorous analysis. Although potentially controversial, we argue that trust in an AI system must be based on natural and intuitive performance properties that should be expected to hold for any {\em reasonable} system, i.e., trust in the system is limited to what can be formally proven about it. Otherwise, it is a purely faith-based system whose trust derives solely from trust in its provider.

\appendix

\section{Uniqueness Proof:}
\label{appendix_uniqueness}
Here we restate the uniqueness theorem and we based the proof from \cite{acmrs}.
\begin{theorem} \textbf{(Uniqueness)}
The result from $\SCCA (A, k)$ is unique, even if there are different sets of shifting vectors ${\matrixsubtensor} \leftarrow \CSA(A, k)$.
\end{theorem}
To facilitate the subsequent theorem and proof, we introduce the following definition.
\begin{definition} For $\pi \subseteq [d]$ and $\Vec{\alpha} \in \mathbb{Z}^{d}$, a vector $\Vec{\alpha}_{\pi} \in \Z^{k}_{>0}$ for $k\leq d$ is considered a $k$-dimensional subvector of $\Vec{\alpha}$ if $\Vec{\alpha}_{\pi} = (\alpha_{\pi_1}, \dots, \alpha_{\pi_k})$. We define the set of those $k$-dimensional subvectors as $V_k(\vec{\alpha})$. 
\label{uniqueness_def}
\end{definition}
Using this definition, we obtain the following proof regarding uniqueness of the recommendation/entry-completion result.
\begin{proof}
For $A' = SCCA(A, k)$, since entry $\vec{\alpha} \in \known{A}$ has $A'(\vec{\alpha} ) = A(\vec{\alpha} )$
by Theorem \ref{uniquess_A'}, we need only prove uniqueness of any completion $\vec{\alpha}\in \unknown{A}$.
From the uniqueness result of Theorem \ref{uniquess_A'}, $\SCCA(A, k)$ admits two distinct shifting vectors $\coeffsubtensorScaled_k$ and $\coeffsubtensorScaled_k'$ that yield the same, unique, $\SCCA(A, k)$. From definition \ref{fully-supported-tensor}, there exists $2^d-1$ vector $\vec{\alpha}' \in H(\vec{\alpha}, \vec{s})$ for some fixed vector $\vec{s}$. 
From Theorem \ref{uniquess_A'}, the shifting vector $\coeffsubtensorScaled_k'$ equals $\coeffsubtensorScaled_k + \coeffsubtensorAdd_k$ is equivalent to
\begin{equation}
        \sum\limits_{i:\vec{\alpha}' \in A_i}\hspace{-4pt} \coeffsubtensorAdd_{k, i}  ~= ~0 \qquad \forall \vec{\alpha}'\in H(\Vec{\alpha}, \Vec{s}).
\end{equation}
We now show that 
\begin{equation}
        A'(\vec{\alpha}) = \sum_{i:\vec{\alpha} \in A_i}\coeffsubtensorScaled_{k,i} = \sum_{i:\vec{\alpha} \in A_i} \coeffsubtensorScaled'_{k,i}
\end{equation}
or equivalently from Theorem \ref{uniquess_A'}
\begin{equation}
     \sum\limits_{i:\vec{\alpha} \in A_i}\hspace{-4pt} \coeffsubtensorAdd_{k, i}  ~= ~0 ~.
\end{equation}
Without loss of generality, we consider the case $d \equiv 0$ (mod $2)$, and the other case can be proven similarly. We define two sets $G_0$ and $G_1$ by the following. Except for $\vec{\alpha}' = \vec{\alpha} + \vec{s}$, we divide $2^d-2$ remaining vectors $\vec{\alpha}'$ into two groups. Then for $m \in \{0, 1\}$, $G_m$ is the set of $\vec{\alpha}' = \vec{\alpha} + \vec{\delta}\cdot \vec{s}$ such that the number of $\delta_i = 0$ equals $m$ modulo 2. We also denote $G_m \cap A_i= \{\vec{\alpha}' \in G_m \mid \vec{\alpha}' \in A_i\}$. Then
\begin{equation}
    \sum_{\vec{\alpha}' \in G_m}~\sum\limits_{i:\vec{\alpha}' \in A_i}\hspace{-8pt} \coeffsubtensorAdd_{k, i} ~~ = ~ 0 ~~\Leftrightarrow ~ \sum_{i}\hspace{-2pt} \coeffsubtensorAdd_{k, i}^{|G_m \cap A_i|} ~~=~ 0 ~.
\end{equation}
By definition \ref{uniqueness_def}, we consider $\vec{v} \in \bigcup_{\Vec{\alpha}' \in H(\vec{\alpha}, \vec{s})} V_{k}(\Vec{\alpha}')$ and a permutation $\pi' \subseteq [k]$ as in definition \ref{subtensors_def} such that $0 \leq k' \leq k$. Then the complement permutation $\pi'^C \in [k] - \pi'$ gives $\Vec{v}'_{\pi'} = \Vec{\alpha}_{\pi'}$ and $\Vec{v}'_{\pi'^C} = (\vec{\alpha} + \vec{s})_{\pi'^C}$. Consider case 1 when $0 < k' < k$, we form a fixed $\vec{\alpha}'$ from $\Vec{v}$ by a new permutation $\pi$ such that $\pi \subseteq [d]$ and $\pi' \subseteq \pi$ by a difference of $l$ elements, making $\pi$ has $k' + l \leq d$ elements. Then the complement permutation $\pi^C \in [d] - \pi$ gives $\Vec{\alpha}'_{\pi} = \Vec{\alpha}_{\pi}$ and $\Vec{\alpha}'_{\pi^C} = (\vec{\alpha} + \vec{s})_{\pi^C}$. 
\\
For $m\in \{0,1\}$, if the number of elements of $\pi$ as $k' + l \equiv m$ (mod $2)$, then $\vec{v}$ forms $\binom{d-k}{l}$ numbers of $\vec{\alpha}'$ that belongs to $G_m$. For subtensor $A_i$ and with the sum in between $0 \leq l \leq d-k$, $|G_m \cap A_i|$ equals $ \sum_{k' + l \equiv m (\text{mod } 2)} \binom{d-k}{l} = 2^{d-k-1}$ for any $m \in \{0, 1\}$. Thus,
\begin{equation} 
   \coeffsubtensorAdd_{k, i}^{|G_1 \cap A_i|}-\coeffsubtensorAdd_{k, i}^{|G_0 \cap A_i|} ~=~ 0 ~.
\end{equation}
If $k' = k$, we encounter the vector $\vec{\alpha}$ when forming $\vec{\alpha}'$. If $k' = 0$, we encounter the vector $\vec{\alpha} + \vec{s}$ when forming $\vec{\alpha}'$. Since we omit $\vec{\alpha}$ and $\vec{\alpha}+\vec{s}$ from $G_0$, $|G_0 \cap A_i| = 2^{d-k-1} - 1$ and $|G_1 \cap A_i| = 2^{d-k-1}$ in either case of $k'$. Thus,
\begin{equation} 
    \coeffsubtensorAdd_{k, i}^{|G_1 \cap A_i|} - \coeffsubtensorAdd_{k, i}^{|G_0 \cap A_i|} ~=~ \coeffsubtensorAdd_{k, i} 
\end{equation}
and therefore
\begin{equation}
    \sum_{i} \coeffsubtensorAdd_{k, i}^{|G_1 \cap A_i|} - \sum_{i} \coeffsubtensorAdd_{k, i}^{|G_0 \cap A_i|} = ~0 ~\Rightarrow \sum\limits_{i:\vec{\alpha} \in A_i}\hspace{-4pt} \coeffsubtensorAdd_{k, i} + \sum\limits_{i:\vec{\alpha}' \in A_i}\hspace{-4pt} \coeffsubtensorAdd_{k, i}\,=  \sum\limits_{i:\vec{\alpha} \in A_i}\hspace{-4pt} \coeffsubtensorAdd_{k, i} ~= 0 ~.
\end{equation}
This equality implies that $A'$ is unchanged, and thus uniquely determined.

\end{proof}

\section{Consesus Ordering Proof}
\label{consensus-proof}
Here we restate the theorem and we based the proof from \cite{acmrs}.
\begin{theorem} Consensus Ordering: 
     Given a tensor $A$ and obtained result $A' = SCCA(A, d-1)$, and permutation vector $\gamma$, $\known{\gamma} \neq \emptyset$, then any completion vector $\vec{\alpha} \in \unknown{\gamma}$ must satisfy $ A'^{(\gamma_a)}(\vec{\alpha}) < A'^{(\gamma_b)}(\vec{\alpha}) \text{ when } a < b.$
\end{theorem}

\begin{proof}
For any vector $\vec{\alpha} \in \known{\gamma}$, the ordering condition from definition $\ref{unanimous}$ gives:
\begin{equation}
    A^{(\gamma_1)}(\Vec{\alpha}) < \dots < A^{(\gamma_i)}(\Vec{\alpha}) < \dots < A^{(\gamma_D)}(\Vec{\alpha})
\label{ordering}
\end{equation}

Given that $A'$ from $\SCCA$ is normalized via $\CSA$, then without loss of generality we can omit the dimension $k$ when using notation (4) from section 2, since $k$ is specified as $k = d-1$. Then the formulaic relationship between $A$ and $A'$ is analogous as follow. Assuming that $S = \left(S_{\{1\}},\cdots S_{\{d\}}\right)$ as $S_{\pi}$ is the shifting by permutation $\pi \in \{\{1\}, \cdots, \{d\}\}$ and we limit $S_{\{d\}} = (\coeffsubtensorScaled_{\gamma_i})_{1\leq i\leq D}$, we obtain the following after the SCCA process
\begin{equation}
\begin{aligned}
      A^{(\gamma_i)}(\Vec{\alpha}) =  \coeffsubtensorScaled_{\gamma_i} + \sum_{j=1}^{d-1} \coeffsubtensorScaled_{ \{j\}, \alpha_j} +  A'^{(\gamma_i)}(\Vec{\alpha}).
\end{aligned}  
\label{multilinear_SCCA}
\end{equation}
Using the CSA result with respect to $(d-1)$-dimensional subtensor in direction $\gamma_i$:
\begin{equation}
    \sum_{\vec{\alpha}\in\known{\gamma}}\hspace{-2pt} A'^{(\gamma_i)}(\Vec{\alpha}) ~=~ 0~.
    \label{normalized_proof}
\end{equation}
Substituting (\ref{multilinear_SCCA}) and (\ref{normalized_proof}) into (\ref{ordering}) and using the fact that $\known{A^{(\gamma_i)}} = \known{\gamma}$ for all $1 \leq i \leq D$ gives

\begin{equation}
   \coeffsubtensorScaled_{\gamma_1} < \dots < \coeffsubtensorScaled_{\gamma_i} < \dots <  \coeffsubtensorScaled_{\gamma_{D}}.
    \label{2.2.11}
\end{equation}
For any vector $\vec{\alpha}' \in\unknown{\gamma}$, the following entry is uniquely determined from theorem \ref{uniqueness_proof},
\begin{equation}
    A'^{(\gamma_i)}(\Vec{\alpha}') ~=~ 
    \coeffsubtensorScaled_{\gamma_i} + \sum_{j=1}^{d-1} \coeffsubtensorScaled_{\{j\}, \alpha'_j},
\end{equation}
and we therefore deduce that $
A'^{(\gamma_1)}(\Vec{\alpha}')
< 
\dots < A'^{(\gamma_i)}(\Vec{\alpha}')  < \dots < A'^{(\gamma_D)}(\Vec{\alpha}').$
\end{proof}

\end{document}